\documentclass[english]{cccconf}

\usepackage[english]{babel}
\usepackage{theorem}
\theoremheaderfont{\itshape\bfseries}
{\theorembodyfont{\itshape}
	\newtheorem{assumption}{\textbf{Assumption}}
	
	\newtheorem{definition}{\textbf{Definition}}
	\newtheorem{theorem}{\textbf{Theorem}}
	
	\newtheorem{remark}{\textbf{Remark}}
	
	\newtheorem{problem}{\textbf{Problem}}
}
\usepackage{graphicx}
\usepackage{newtxtext,newtxmath}
\usepackage{amsmath}
\usepackage{amsfonts}
\usepackage{mathrsfs}
\usepackage{xcolor}
\usepackage{graphicx}
\usepackage{tcolorbox}
\usepackage{mdframed}
\usepackage{booktabs}
\usepackage{caption}
\usepackage{supertabular}

\usepackage{fullpage}
\usepackage{times}
\usepackage{fancyhdr,graphicx,amsmath,amssymb}
\usepackage[ruled,vlined]{algorithm2e}

\newcommand{\T}{^{\mbox{\tiny T}}}
\newcommand{\R}{\mathbb{R}}

\newcommand{\eps}{\varepsilon}
\newenvironment{Protocol}[1][htb]
{% Update algorithm name
	\begin{algorithm}[#1]%
	}{\end{algorithm}}
\let\leq\leqslant
\let\geq\geqslant

\newenvironment{proof}[1][Proof]%
{\par\noindent\textit{#1:\ }}%
{\hspace*{\fill} \rule{6pt}{6pt}}
\newenvironment{proof*}[1][Proof]%
{\par\noindent\textit{#1:\ }}{}

\DeclareMathOperator{\sat}{sat}
\DeclareMathOperator{\sgn}{sgn}
\usepackage{tikz}
\usetikzlibrary{shapes,calc,arrows,patterns,decorations.pathmorphing
	,decorations.markings}
\usetikzlibrary{arrows.meta}

\newenvironment{system}[1]%
{\setlength{\arraycolsep}{0.5mm}\left\{ \; \begin{array}{#1}}%
	{\end{array} \right.}
\newenvironment{system*}[1]%
{\setlength{\arraycolsep}{0.5mm} \begin{array}{#1}}%
	{\end{array}}

\begin{document}

\title{\LARGE \textbf{Scale-free Linear Observer-based Protocol Design for Global Regulated State Synchronization of Homogeneous Multi-agent Systems with Non-introspective Agents Subject to Input Saturation}}
	\author{ Zhenwei Liu\aref{neu}, Donya Nojavanzadeh\aref{wsu}, Ali Saberi\aref{wsu},
	Anton A. Stoorvogel\aref{ut}}

\affiliation[neu]{College of Information Science and
	Engineering, Northeastern University, Shenyang 110819, China
		\email{liuzhenwei@ise.neu.edu.cn}}
\affiliation[wsu]{School of Electrical Engineering and Computer
	Science, Washington State University, Pullman, WA 99164, USA
	\email{donya.nojavanzadeh@wsu.edu; saberi@wsu.edu}}
\affiliation[ut]{Department of Electrical Engineering,
	Mathematics and Computer Science, University of Twente, Enschede, The Netherlands
	\email{A.A.Stoorvogel@utwente.nl}}

\maketitle

\begin{abstract}
   This paper studies global regulated state synchronization of homogeneous networks of non-introspective agents in presence of input saturation. We identify three classes of agent models which are neutrally stable, double-integrator, and mixed of double-integrator, single-integrator and neutrally stable dynamics. A \textit{scale-free linear observer-based} protocol design methodology is developed based on localized information exchange among neighbors where the reference trajectory is given by a so-called exosystem which is assumed to be globally reachable. Our protocols do not need any knowledge about the communication network topology and the spectrum of associated Laplacian matrix. Moreover, the proposed protocol is scalable and is designed based on only knowledge of agent models and achieves synchronization for any communication graph with arbitrary number of agents.
  \end{abstract}

\keywords{Multi-agent systems, Global regulated state synchronization, Scale-free protocol design}

\footnotetext{This work is supported by Nature Science
	Foundation of Liaoning Province under Grant 2019-MS-116.}

\section{Introduction}

The synchronization problem of networks consisting of linear or nonlinear agents has become a hot topic among researchers during the past decade due to the wide potential for
applications in several areas such as automotive vehicle control,
satellites/robots formation, sensor networks, and so on. The
objective of synchronization is to secure asymptotic agreement on a
common state or output trajectory by control
protocols with local communication information, see for
instance the books \cite{ren-book} and \cite{wu-book} or the
survey paper \cite{saber-murray3}.

Generally, synchronization of multi-agent system (MAS) includes two main types, state and output synchronization. Because the state synchronization inherently requires homogeneous networks
(i.e. agents which have identical dynamics), most work in synchronization for MAS focused on
state synchronization of homogeneous networks. State synchronization based on diffusive \emph{full-state coupling} has been studied where the agent dynamics progress from single- and double-integrator (e.g.  \cite{saber-murray2}, \cite{ren}) to more general dynamics (e.g. \cite{scardovi-sepulchre}, \cite{tuna1},
\cite{wieland-kim-allgower}). State synchronization based on
diffusive \emph{partial-state coupling} has also been considered, including static design (\cite{liu-zhang-saberi-stoorvogel-auto} and \cite{liu-zhang-saberi-stoorvogel-ejc}), dynamic design (\cite{kim-shim-back-seo}, \cite{seo-back-kim-shim-iet}, \cite{seo-shim-back}, \cite{su-huang-tac},
\cite{tuna3}), and the design based on localized information exchange (\cite{chowdhury-khalil} and \cite{scardovi-sepulchre}). Solvability conditions are studied for general case of full and partial-state coupling in \cite{stoorvogel-saberi-zhang-auto2017}, \cite{stoorvogel-saberi-zhang-liu}. Recently, scale-free collaborative protocol designs are developed for continuous-time heterogeneous MAS \cite{donya-liu-saberi-stoorvogel-ACC2020} and for homogeneous MAS subject to actuator saturation \cite{liu-saberi-stoorvogel-donya-cdc2019} and subject to input delays \cite{liu-donya-dmitri-saberi-stoorvogel-arxiv-inputdelay-con,liu-donya-dmitri-saberi-stoorvogel-arxiv-inputdelay-dis}.  

Meanwhile, if the agents have absolute measurements of
their own dynamics in addition to relative information from the network, they are said to be introspective, otherwise, they are called non-introspective. There exist some results about these two types of agents, for example, introspective agents (\cite{kim-shim-seo,yang-saberi-stoorvogel-grip-journal}, etc), and non-introspective agents (\cite{grip-yang-saberi-stoorvogel-automatica, wieland-sepulchre-allgower}, etc).

On the other hand, it is worth to note that actuator saturation is pretty common and indeed is ubiquitous in engineering
applications. Many researchers have tried to establish (semi) global state and output synchronization results for multi-agent system (MAS) in the presence of input saturation. Compared with semi-global results, global synchronization can work for any initial condition set, and thus it has wider applications and attracts more attention.  
Global synchronization for neutrally stable
agents has
been studied by \cite{meng-zhao-lin-2013} (continuous-time) and
\cite{yang-meng-dimarogonas-johansson} (discrete-time) for either undirected or detailed balanced graph. Then, global synchronization via static protocols for MAS with partial state coupling and linear general dynamics is developed in \cite{liu-saberi-stoorvogel-zhang-ijrnc}. Reference \cite{li-xiang-wei} provides the design which can deal with networks that are not
detailed balanced but intrinsically requires the agents to be single integrator. Similar scenarios also can be found in \cite{fu-wen-yu-ding} (finite-time consensus), and \cite{yi-yang-wu-johansson-auto} (event-triggered control).

In this paper, we design \textbf{scale-free linear observer-based} dynamic protocols to achieve global regulated state synchronization for homogeneous
networks of non-introspective agents in presence of input
saturation utilizing localized information exchange among the neighbors. The contributions of this paper are stated as follow.
\renewcommand\labelitemi{{\boldmath$\bullet$}}
\begin{itemize}
	\item We develop scale-free linear observer-based dynamic protocols for MAS with non-introspective agents and for three classes of agent models which are neutrally stable, double-integrator, and mixed of double-integrator, single-integrator and neutrally stable dynamics and for both networks with full- and partial-state coupling. Moreover, the proposed linear protocols have infinite gain margins.
	\item  Linear observer-based protocol designs are
	scale-free and do not need any information about communication network. In other words, the proposed protocols work for any MAS with any communication graph with arbitrary number of agents.
\end{itemize}

\subsection*{Notations and definitions}
Given a matrix $A\in \mathbb{R}^{m\times n}$, $A\T$ denotes the transpose of $A$ and $\|A\|$ denotes the induced 2-norm of $A$. For a vector $x\in \mathbb{R}^q$, $\|x\|$ denotes the 2-norm of $x$ respectively. A square matrix $A$ is said to be Hurwitz stable if all its eigenvalues are in the open left half complex plane.  $A\otimes B$ depicts the
Kronecker product between $A$ and $B$. $I_n$ denotes the
$n$-dimensional identity matrix and $0_n$ denotes $n\times n$ zero
matrix; sometimes we drop the subscript if the dimension is clear from
the context.

% and for a vector signal $v$, we denote the $\mathscr{L}_1$, $\mathscr{L}_2$, and $\mathscr{L}_\infty$ norm by $\| v \|_1$, $\| v \|_2$ and $\| v \|_\infty$
%We denote by
%$\diag\{A_1,\ldots, A_N \}$, a block-diagonal matrix with
%$A_1,\ldots,A_N$ as its diagonal elements.

To describe the information flow among the agents we associate a \emph{weighted graph} $\mathcal{G}$ to the communication network. The weighted graph $\mathcal{G}$ is defined by a triple
$(\mathcal{V}, \mathcal{E}, \mathcal{A})$ where
$\mathcal{V}=\{1,\ldots, N\}$ is a node set, $\mathcal{E}$ is a set of
pairs of nodes indicating connections among nodes, and
$\mathcal{A}=[a_{ij}]\in \mathbb{R}^{N\times N}$ is the weighted adjacency matrix with non negative elements $a_{ij}$. Each pair in $\mathcal{E}$ is called an \emph{edge}, where
$a_{ij}>0$ denotes an edge $(j,i)\in \mathcal{E}$ from node $j$ to
node $i$ with weight $a_{ij}$. Moreover, $a_{ij}=0$ if there is no
edge from node $j$ to node $i$. We assume there are no self-loops,
i.e.\ we have $a_{ii}=0$. A \emph{path} from node $i_1$ to $i_k$ is a
sequence of nodes $\{i_1,\ldots, i_k\}$ such that
$(i_j, i_{j+1})\in \mathcal{E}$ for $j=1,\ldots, k-1$. A \emph{directed tree} is a subgraph (subset
of nodes and edges) in which every node has exactly one parent node except for one node, called the \emph{root}, which has no parent node. The \emph{root set} is the set of root nodes. A \emph{directed spanning tree} is a subgraph which is
a directed tree containing all the nodes of the original graph. If a directed spanning tree exists, the root has a directed path to every other node in the tree.  

For a weighted graph $\mathcal{G}$, the matrix
$L=[\ell_{ij}]$ with
\[
\ell_{ij}=
\begin{system}{cl}
\sum_{k=1}^{N} a_{ik}, & i=j,\\
-a_{ij}, & i\neq j,
\end{system}
\]
is called the \emph{Laplacian matrix} associated with the graph
$\mathcal{G}$. The Laplacian matrix $L$ has all its eigenvalues in the
closed right half plane and at least one eigenvalue at zero associated
with right eigenvector $\textbf{1}$ \cite{royle-godsil}. Moreover, if the graph contains a directed spanning tree, the Laplacian matrix $L$ has a single eigenvalue at the origin and all other eigenvalues are located in the open right-half complex plane \cite{ren-book}.

\section{Problem Formulation}

Consider a MAS consisting of $N$ identical dynamic agents with input saturation:
\begin{equation}\label{eq1}
\begin{cases}
\dot{x}_i=Ax_i+B\sigma(u_i),\\
y_i=Cx_i,
\end{cases}
\end{equation}
where $x_i\in\mathbb{R}^n$, $y_i\in\mathbb{R}^q$ and
$u_i\in\mathbb{R}^m$ are the state, output, and the input of agent 
$i=1,\ldots, N$, respectively. Meanwhile,
\[
\sigma(v)=\begin{pmatrix} 
\sat(v_1) \\ \sat(v_2) \\ \vdots \\ \sat(v_m)
\end{pmatrix}\text{ where }
v=\begin{pmatrix} 
v_1 \\ v_2 \\ \vdots \\ v_m
\end{pmatrix} \in \R^m
\]
with $\sat(w)$ is the standard saturation function:
\[
\sat(w)=\sgn(w)\min(1,|w|).
\]

%\begin{assumption}\label{Aass}
%	Assume agents are at most weakly unstable, namely, all eigenvalues of $A$ are in the closed left half plane. Moreover, let $(A, B, C)$ be stabilizable and detectable.
%\end{assumption}

The network provides agent $i$ with the following information,
\begin{equation}\label{eq2}
\zeta_i=\sum_{j=1}^{N}a_{ij}(y_i-y_j),
\end{equation}
where $a_{ij}\geq 0$ and $a_{ii}=0$. This communication topology of
the network can be described by a weighted graph $\mathcal{G}$ associated with \eqref{eq2}, with
the $a_{ij}$ being the coefficients of the weighted adjacency matrix
$\mathcal{A}$. In terms of the coefficients of the associated
Laplacian matrix $L$, $\zeta_i$ can be rewritten as
\begin{equation}\label{zeta_l}
\zeta_i = \sum_{j=1}^{N}\ell_{ij}y_j.
\end{equation}
We refer to \eqref{zeta_l} as \emph{partial-state coupling} since only part of
the states are communicated over the network. When $C=I$, it means all states are shared over the network and we call it \emph{full-state coupling}. 

%Then, the original agents are expressed as
%\begin{equation}\label{neq1}
%\dot{x}_i=Ax_i+B\sigma(u_i)
%\end{equation}
%%and
%%\begin{equation}\label{nsolu-cond}
%%\dot{x}_r  = A x_r
%%\end{equation}
%and $\zeta_i$ is rewritten as
%\[
%\zeta_i = \sum_{j=1}^{N}\ell_{ij}x_j.
%\]

We also introduce a localized
information exchange among neighbors. In particular, each agent 
$i=1,\ldots, N$ has access to a localized information denoted by
$\hat{\zeta}_i$, of the form
\begin{equation}\label{eqa1}
\hat{\zeta}_i=\sum_{j=1}^Na_{ij}(\xi_i-\xi_j)
\end{equation}
where $\xi_j\in\mathbb{R}^n$ is a variable produced internally by agent $j$ and to be defined in next sections.

%Generally, for heterogeneous agents, it has been
%shown in \cite{wieland-sepulchre-allgower,grip-saberi-stoorvogel3}
%that we basically need to consider regulated state synchronization
%where the objective of the agents is to ensure that their state
%asymptotically tracks a reference trajectory generated by a so-called exosystem.
%Although we consider homogeneous MAS, we will study regulated state synchronization in this paper.
%Although our paper considers homogeneous MAS, the adaptation required
%to handle the lack of information about the network introduces a
%certain level of heterogeneity. Therefore, we also consider regulated
%state synchronization in this paper.

In this paper, we consider regulated state synchronization where state of agents converge to a priori given trajectory $x_r$ generated by a so-called exosystem
\begin{equation}\label{solu-cond}
\dot{x}_r  = A x_r, \quad y_r=Cx_r.
\end{equation}
with $x_r\in\R^n$.
Clearly, we need some level of
communication between the exosystem and the agents.  We assume that a nonempty
subset $\mathscr{C}$ of the agents have access to their
own output relative to the output of the exosystem.  Specially, each
agent $i$ has access to the quantity
\begin{equation}\label{elp}
\psi_i=\iota_{i}(y_i-y_r),\qquad
\iota_i=
\begin{cases}
1, & i\in \mathscr{C},\\
0, & i\notin \mathscr{C}.
\end{cases}
\end{equation}
Combined with \eqref{eq2}, we have the following network
exchange
\begin{equation}\label{zeta_l-0}
\bar{\zeta}_i = \sum_{j=1}^{N}a_{ij}(y_i-y_j)+\iota_{i}(y_i-y_r).
\end{equation}
$\bar{\zeta}_i$, as defined in above, can be rewritten in terms of the coefficients of a so-called expanded Laplacian matrix $\bar{L}=L+diag\{\iota_i\}=[\bar{\ell}_{ij}]_{N \times N}$ as
\begin{equation}\label{zeta_l-n}
\bar{\zeta}_i=\sum_{j=1}^{N}\bar{\ell}_{ij}(y_j-y_r).
\end{equation}
%Meanwhile, for full-state coupling case  \eqref{zeta_l-n} will be changed as
%\begin{equation}\label{zeta_l-nn}
%\bar{\zeta}_i=\sum_{j=1}^{N}\bar{\ell}_{ij}(x_j-x_r).
%\end{equation}
Note that $\bar{L}$ is not a regular Laplacian matrix associated to the graph, since the sum of its rows need not be zero. We know that all the eigenvalues of $\bar{L}$, have positive real parts. In particular matrix $\bar{L}$ is invertible.

To guarantee that each agent gets the information from the exosystem, we need to make sure that there exists a path from node set $\mathscr{C}$ to each node.  Therefore, we define the following set of graphs.
\begin{definition}\label{def_rootset}
	Given a node set $\mathscr{C}$, we denote by $\mathbb{G}_{\mathscr{C}}^N$ the set of all graphs with $N$ nodes containing the node set $\mathscr{C}$, such that every node of the network graph $\mathcal{G}\in\mathbb{G}_\mathscr{C}^N$ is a member of a directed tree
	which has its root contained in the node set $\mathscr{C}$. We will refer to the node set $\mathscr{C}$ as root set.
\end{definition}

\begin{remark}
	Note that Definition \ref{def_rootset} does not require necessarily the existence of directed spanning tree. If the root of the trees
	belongs to the set $\mathscr{C}$, this means all the agents of the
	network will have access to the information of the exosystem,
	i.e. we do not need necessarily the existence of the spanning tree.
\end{remark}

Next, we formulate \textbf{scalable} global regulated state synchronization problem with linear protocols.

\begin{problem}\label{prob2}
	Consider a MAS described by \eqref{eq1} and \eqref{zeta_l-n} and the
associated exosystem \eqref{solu-cond}. Let a set of nodes
$\mathscr{C}$ be given which defines the set
$\mathbb{G}_{\mathscr{C}}^N$.
	
	The \textbf{scalable global regulated state synchronization problem based on localized information exchange} of a MAS is to find, if possible, a linear observer-based dynamic protocol for each agent $i\in\{1,\hdots,N\}$, using only
	knowledge of agent model, i.e. $(A,B,C)$, of the form:
	\begin{equation}\label{protoco3}
	\begin{system}{cl}
	\dot{x}_{c,i}&=A_{c} x_{c,i}+B_{c}{\sigma(u_i)}+C_{c} \bar{\zeta}_i+D_{c} \hat{\zeta}_i,\\
	u_i&=F_cx_{c,i}
	\end{system}
	\end{equation}
	where $\hat{\zeta}_i$ is defined in \eqref{eqa1} with $\xi_i=H_{c}x_{i,c}$, and $x_{c,i}\in\R^{n_c}$, such that regulated state synchronization
	\begin{equation}\label{synch_org}
	\lim_{t\to \infty} (x_i-x_j)=0 \quad \text{ for all } i,j \in {1,...,N}
	\end{equation}
	 is achieved for any $N$ and any graph $\mathcal{G}\in \mathbb{G}_{\mathscr{C}}^N$, and for all initial conditions of the agents $x_i(0) \in \mathbb{R}^n$, all initial conditions of the exosystem $x_r(0) \in \mathbb{R}^n$, and all initial conditions of the protocols $x_{c,i}(0) \in \mathbb{R}^{n_c}$.
\end{problem}

\begin{remark}
	In the case of full-state coupling, matrix $C=I$ and we refer to Problem \ref{prob2} as \textbf{scalable global regulated state synchronization problem based on localized information exchange for MAS with full-state coupling}.
\end{remark}

\section{MAS with Neutrally Stable Agents}

In this section, we will consider the scalable global regulated state synchronization problem for a
MAS consisting of neutrally stable agents with input saturation for both networks with full- and partial-state coupling. We make the following assumption on agent models.
\begin{assumption}\label{Aass1}
	We assume that $(A, B, C)$ is controllable and observable.
	Moreover, $A$ is neutrally stable, i.e., all the eigenvalues of $A$ are in the closed left half plane and those eigenvalues on the imaginary axis, if any, are semi-simple.
\end{assumption}

\subsection{Full-state coupling}

In this subsection we consider MAS with full-state coupling.

\vspace{0.5cm}

\begin{Protocol}[h]	\caption{{Full-state coupling} \label{p1fsc}}	
The following protocol is designed for each agent
$i\in\{1,\ldots,N\}$,
\begin{equation}\label{pscpm0}
\begin{system}{cll}
\dot{\chi}_i &=& A\chi_i+B\sigma(u_i)+\bar{\zeta}_i-\hat{\zeta}_i-\iota_i\chi_i \\
u_i &=& -\rho B\T P\chi_i,
\end{system}
\end{equation}
where $\rho>0$ is a parameter with arbitrary positive value and $P>0$ satisfies
\begin{equation}\label{condforneut}
PA+A\T P\leq 0
\end{equation}
since $A$ satisfies Assumption \ref{Aass1}.
The agents communicate $\xi_i$ which is chosen as $\xi_i=\chi_i$, therefore each agent has access to the following information:
\begin{equation}\label{info1}
\hat{\zeta}_i=\sum_{j=1}^Na_{ij}(\chi_i-\chi_j).
\end{equation}
while $\bar{\zeta}_i$ is defined by \eqref{zeta_l-n}.
\end{Protocol}

\vspace{0.5cm}

We have following theorem for scalable global regulated state synchronization based on localized information exchange for MAS with full-state coupling and neutrally stable agent models.

\begin{theorem}\label{mainthm0}
	Consider a MAS with neutrally stable agents described by \eqref{eq1} where $C=I$, satisfying Assumption \ref{Aass1}, and the associated exosystem
	\eqref{solu-cond}. Let a set of nodes $\mathscr{C}$ be given which
	defines the set $\mathbb{G}_{\mathscr{C}}^N$. Let the associated
	network communication be given by \eqref{zeta_l-n}. 
	
	Then, the scalable global regulated state synchronization problem based on localized information exchange for MAS with full-state coupling as stated in Problem
	\ref{prob2} is solvable. In particular, for any given $\rho>0$, the
	dynamic protocol \eqref{pscpm0} solves the regulated state
	synchronization problem for any $N$ and any graph
	$\mathcal{G}\in\mathbb{G}_{\mathscr{C}}^N$. 
\end{theorem}

\begin{proof}[Proof of Theorem \ref{mainthm0}]
	Firstly, by defining $\tilde{x}_i=x_i-x_r$ and $e_i=\tilde{x}_i-\chi_i$ we have
	\begin{equation*}
	\begin{system*}{l}
	\dot{\tilde{x}}_i=A\tilde{x}_i+ B\sigma (u_i),\\
	\dot{e}_i=Ae_i-\sum_{j=1}^{N}\bar{\ell}_{ij}e_j,\\
	u_i=-\rho B\T P(\tilde{x}_i-e_i)
	\end{system*}	
	\end{equation*}
	 Then, let
	\[
	\tilde{x}=\begin{pmatrix}
	\tilde{x}_1\\\vdots\\\tilde{x}_N
	\end{pmatrix} ,
	u=\begin{pmatrix}
	u_1\\\vdots\\u_N
	\end{pmatrix}, 
	e=\begin{pmatrix}
	e_1\\\vdots\\e_N
	\end{pmatrix}, \text{ and } 
	\sigma(u)=\begin{pmatrix}
	\sigma(u_1)\\\vdots\\\sigma(u_N)
	\end{pmatrix}
	\]
	then we have the following closed-loop system
	\begin{equation}\label{closedls1}
	\begin{system}{l}
	\dot{\tilde{x}}=(I\otimes A)\tilde{x}+ (I\otimes B)\sigma (u),\\
	\dot{e}=(I\otimes A-\bar{L}\otimes I)e,\\
	u=-\rho (I\otimes  B\T P)(\tilde{x}-e).
	\end{system}	
	\end{equation}
	
	Since all eigenvalues of $\bar{L}$ have positive real part, we have
	\begin{equation}\label{boundapl}
	(T\otimes I)(I\otimes A-\bar{L}\otimes I)(T^{-1}\otimes I)=I\otimes A-\bar{J}\otimes I
	\end{equation}
	for a non-singular transformation matrix $T$, where
	%\[
	%\bar{\Lambda}=\begin{pmatrix}
	%{\lambda}_1&&\\
	%&\ddots&\\
	%&&{\lambda}_{N-1}
	%\end{pmatrix}.
	%\]
	\eqref{boundapl}  is upper triangular Jordan form with $A-\lambda_i I$ for $i=1,\cdots,N$ on the diagonal. Since the agents are neutrally stable, i.e. all eigenvalues of $A$ are in the closed left half plane, $A-\lambda_i I$ is stable. Therefore, all eigenvalues of $I\otimes A-\bar{L}\otimes I$ have negative real part.
	
	Then, we choose the following Lyapunov function
	\begin{equation}
	V=\tilde{x}\T (I\otimes P)\tilde{x}+e\T \bar{P}e
	\end{equation}
	where $P>0$ satisfies condition \eqref{condforneut} and $\bar{P}>0$ satisfies
	\begin{equation}\label{condfe}
	\bar{P}(I\otimes A-\bar{L}\otimes I)+(I\otimes A-\bar{L}\otimes I)\T\bar{P}\leq -(1+\rho \|B\T P\|^2) I
	\end{equation}
	Thus, we have
	\begin{align*}
	\nonumber\frac{dV}{dt}=&\tilde{x}\T I\otimes (PA+A\T P)\tilde{x}+ 2\tilde{x}\T(I\otimes PB)\sigma (u)\\
	\nonumber&+e\T[\bar{P}(I\otimes A-\bar{L}\otimes I)+(I\otimes A-\bar{L}\otimes I)\T\bar{P}]e\\
	\leq&-2\rho^{-1}u\T\sigma (u)+2e\T (I\otimes PB)\sigma (u)\\
	&- (1+\rho \|B\T P\|^2)e\T e\\
%	\leq&-2\rho^{-1}u\T\sigma (u)+\rho^{-1}\sigma\T(u)\sigma (u)\\
%	&+e\T (\rho I\otimes PBB\T P-(1+\rho \|B\T P\|^2) I) e\\
	\leq&-2\rho^{-1}u\T\sigma (u)+\rho^{-1}\sigma\T(u)\sigma (u)-\|e\|^2
	\end{align*}
	
	Since $u_i^k\sigma (u_i^k)=|u_i^k||\sigma (u_i^k)|\geq|\sigma (u_i^k)|^2$ ($u_i^k$ is $k$th element of $u_i$, $k=1,\cdots,n$), we have $
	-2u\T\sigma (u)+\sigma\T(u)\sigma (u)\leq0$.	
	Thus, we obtain
	$
	\frac{dV}{dt}\leq 0$.
	
	Meanwhile, we note that $\frac{dV}{dt}=0$ when $I\otimes (PA+A\T P)\tilde{x}=0$, $(I\otimes B\T P)\tilde{x}=0$, and $e=0$ based on \eqref{condfe}.
	Thus in this case, $\tilde{x}$ is the solution of the dynamics $\dot{\tilde{x}}=(I\otimes A)\tilde{x}$.

	Let $S$ be a matrix such that $A+BS$ is Hurwitz stable. Then we have
	\[
	(I\otimes P)\dot{\tilde{x}}=I\otimes (PA-S\T B\T P)\tilde{x}=-I\otimes (A\T+S\T B\T) (I\otimes P)\tilde{x}
	\]
	since $[I\otimes (PA)]\tilde{x}=-[I\otimes (A\T P)]\tilde{x}$ and $[I\otimes (S\T B\T P)]\tilde{x}=0$. Because $A\T+S\T B\T $ is Hurwitz stable, we have $(I\otimes P)\tilde{x}$ is exponentially growing which contradicts with $\dot{\tilde{x}}=(I\otimes A)\tilde{x}$. It means that $\tilde{x}=0$ is the solution of the above dynamics when $P>0$. Thus, the invariance set $\{(\tilde{x},e): \dot{V}(\tilde{x},e)=0\}$ contains no trajectory of the system except the trivial trajectory $(\tilde{x},e)=(0,0)$.	
	Therefore, system \eqref{closedls1} is globally asymptotically stable based on LaSalle's invariance principle. 
	It means we have $\tilde{x}\to 0$ and $e\to0$ when $t\to \infty$. Thus we obtain $x_i\to x_r$ as $t\to \infty$, which proves our result.
\end{proof}

\subsection{Partial-state coupling}
In this subsection we consider MAS with partial-state coupling.

\vspace{0.5cm}

\begin{Protocol}[h]	\caption{{partial-state coupling} \label{p2psc}}
The following protocol is designed for each agent
$i\in\{1,\ldots,N\}$,
\begin{equation}\label{pscpm02}
\begin{system}{cll}
\dot{\hat{x}}_i &=& A\hat{x}_i+B\hat{\zeta}_{i2}+F(\bar{\zeta}_i-C\hat{x}_i)+\iota_iB\sigma(u_i) \\
\dot{\chi}_i &=& A\chi_i+B\sigma(u_i)+\hat{x}_i-\hat{\zeta}_{i1}-\iota_{i}\chi_i \\
u_i &=& -\rho B\T P\chi_i,
\end{system}
\end{equation}
where $F$ is a design matrix such that $A-FC$ is Hurwitz stable, $\rho>0$ is a parameter with arbitrary positive value, and $P$ satisfies \eqref{condforneut}.
In this protocol, the agents communicate $\xi_i=\begin{pmatrix}
\xi_{i1}\T,&\xi_{i2}\T
\end{pmatrix}\T=\begin{pmatrix}
\chi_i\T,&\sigma\T(u_i)
\end{pmatrix}\T$, i.e. each agent has access to localized information $\hat{\zeta}_i=\begin{pmatrix}
\hat{\zeta}_{i1}\T,&\hat{\zeta}_{i2}\T
\end{pmatrix}\T$, where $\hat{\zeta}_{i1}$ and $\hat{\zeta}_{i2}$ are defined as
\begin{equation}\label{add_1}
\hat{\zeta}_{i1}=\sum_{j=1}^Na_{ij}(\chi_i-\chi_j),
\end{equation}
and
\begin{equation}\label{add_3}
\hat{\zeta}_{i2}=\sum_{j=1}^{N}a_{ij}(\sigma(u_i)-\sigma(u_j)),
\end{equation}
while $\bar{\zeta}_i$ is defined via \eqref{zeta_l-n}.
\end{Protocol}

\vspace{0.5cm}

Then, we have the following theorem for scalable global regulated state synchronization based on localized information exchange for MAS with partial-state coupling and neutrally stable agent models.

\begin{theorem}\label{mainthm02}
	Consider a MAS with neutrally stable agents described by \eqref{eq1} satisfying Assumption \ref{Aass1}, and the associated exosystem
	\eqref{solu-cond}. Let a set of nodes $\mathscr{C}$ be given which
	defines the set $\mathbb{G}_{\mathscr{C}}^N$. Let the associated
	network communication be given by \eqref{zeta_l-n}. 
	
	Then, the scalable global regulated state synchronization problem based on localized information exchange for MAS with partial-state coupling as stated in Problem
	\ref{prob2} is solvable. In particular, for any given $\rho>0$, the
	dynamic protocol \eqref{pscpm02} solves the scalable regulated state
	synchronization problem for any $N$ and any graph
	$\mathcal{G}\in\mathbb{G}_{\mathscr{C}}^N$. 
\end{theorem}

\begin{proof}[Proof of Theorem \ref{mainthm02}]
	Similar to the proof of Theorem \ref{mainthm0}, we have 
	the matrix expression of closed-loop system 
	\begin{equation}\label{newsystem}
	\begin{system}{l}
	\dot{\tilde{x}}=(I\otimes A)\tilde{x}+ (I\otimes B)\sigma (u)\\
	\dot{e}=(I\otimes A-\bar{L}\otimes I)e+\bar{e}\\
	\dot{\bar{e}}=I\otimes (A-FC)\bar{e}\\	
	u=-\rho (I\otimes B\T P) (\tilde{x}-e)
	\end{system}
	\end{equation}
	by $e=\tilde{x}-\chi$, and $\bar{e}=(\bar{L}\otimes I)\tilde{x}-\hat{x}$.
	
	Then, choose the following Lyapunov function
	\[
	V=\tilde{x}\T (I\otimes P)\tilde{x}+\begin{pmatrix}
	e\\\bar{e}
	\end{pmatrix}\T\tilde{P}\begin{pmatrix}
	e\\\bar{e}
	\end{pmatrix}
	\]
	where $P>0$ satisfies \eqref{condforneut} and $\tilde{P}>0$ satisfies 
	\begin{equation}
	\tilde{P}\bar{A}+\bar{A}\T\tilde{P}\leq -(\rho \|B\T P\|^2+1)I
	\end{equation}
	with
	\[
	\bar{A}=\begin{pmatrix}
	I\otimes A-\bar{L}\otimes I&I\\0&I\otimes (A-FC)
	\end{pmatrix}.
	\]
	
	Similar to Theorem \ref{mainthm0}, we can obtain the synchronization result $x_i\to x_r$ as $t\to \infty$.	
\end{proof}	

\vspace{1cm}

\section{MAS with Double-integrator Agents}

In this section, we will consider scalable global regulated state synchronization problem for MAS consisting of double-integrator agents with input saturation for both networks with full and partial-state coupling.

\subsection{Full-state coupling}
In this subsection, we design dynamic protocols for MAS with full-state coupling and double-integrator agent models. First, for agents \eqref{eq1} with double integrator models, we have
\begin{equation}\label{doubleAB}
A=\begin{pmatrix}
0&I_m\\0&0
\end{pmatrix}, B=\begin{pmatrix}
0\\I_m
\end{pmatrix}
\end{equation}
where $A \in \mathbb{R}^{2m\times 2m}$ and $B \in \mathbb{R}^{2m\times m}$.  Then, we choose matrix $K=\begin{pmatrix}
K_1&K_2
\end{pmatrix}$ such that $K_i \in \mathbb{R}^{m\times m}$, $i=1,2$ are arbitrary negative definite matrices.
	%i.e. \eqref{eq1} can also be expressed as
	%\begin{equation}
	%\begin{system}{rl}
	%\dot{x}_i^{I}&=x_i^{II}\\
	%\dot{x}_i^{II}&=\sigma(u_i)\\
	%y_i&=Cx_i
	%\end{system}
	%\end{equation}
	%with $x_i=\begin{pmatrix}
	%\left(x_i^{I}\right)\T&\left(x_i^{II}\right)\T
	%\end{pmatrix}\T$.
	
%	\begin{align}
%	&K_1=K_1\T<0\label{condd1}\\
%	&K_2+K_2\T<-\eps I\label{condd2}
%	\end{align}
%	with $\eps>0$ is a sufficiently small number.

\vspace{0.5cm}

\begin{Protocol}[h]	\caption{{full-state coupling} \label{p3fsc}}
The following protocol is designed for each agent
$i\in\{1,\ldots,N\}$,
\begin{equation}\label{pscpd1}
\begin{system}{cll}
\dot{\chi}_i &=& A\chi_i+B\sigma(u_i)+\bar{\zeta}_i-\hat{\zeta}_i-\iota_i\chi_i \\
u_i &=&\rho K\chi_i,
\end{system}
\end{equation}
where $\rho>0$ is a parameter with arbitrary positive value, and $\hat{\zeta}_i$ and $\bar{\zeta}_i$ are defined by \eqref{info1} and \eqref{zeta_l-n}, respectively.
\end{Protocol}

\vspace{0.5cm}

 We have the following theorem for scalable global regulated state synchronization problem based on localized information exchange for MAS with full-state coupling and double-integrator agent models.
 
\begin{theorem}\label{mainthm3}
	Consider a MAS described by \eqref{eq1} with \eqref{doubleAB} and $C=I$, and the associated exosystem
	\eqref{solu-cond}. Let a set of nodes $\mathscr{C}$ be given which
	defines the set $\mathbb{G}_{\mathscr{C}}^N$. Let the associated
	network communication be given by \eqref{zeta_l-n}. 
	
	Then, the scalable global regulated state synchronization problem based on localized information exchange for MAS with full-state coupling as stated in Problem
	\ref{prob2} is solvable. In particular, for any given $\rho>0$, the
	dynamic protocol \eqref{pscpd1} solves the regulated state
	synchronization problem for any $N$ and any graph
	$\mathcal{G}\in\mathbb{G}_{\mathscr{C}}^N$. 
\end{theorem}

\begin{proof}[Proof of Theorem \ref{mainthm3}]
	Firstly, similar to Theorem \ref{mainthm0}, we have
	\begin{equation*}
	\begin{system*}{l}
	\dot{\tilde{x}}_i=A\tilde{x}_i+ B\sigma (u_i),\\
	\dot{e}_i=Ae_i-\sum_{j=1}^{N}\bar{\ell}_{ij}e_j,\\
	u_i=\rho K(\tilde{x}_i-e_i)
	\end{system*}	
	\end{equation*}
	by $\tilde{x}_i=x_i-x_r$ and $e_i=\tilde{x}_i-\chi_i$. Then, let
	\[
	\tilde{x}=\begin{pmatrix}
	\tilde{x}_1\\\vdots\\\tilde{x}_N
	\end{pmatrix},
	u=\begin{pmatrix}
	u_1\\\vdots\\u_N
	\end{pmatrix},
	e=\begin{pmatrix}
	e_1\\\vdots\\e_N
	\end{pmatrix}, \text{ and } 
	\sigma(u)=\begin{pmatrix}
	\sigma(u_1)\\\vdots\\\sigma(u_N)
	\end{pmatrix}
	\]
	where $
	\tilde{x}_i=\begin{pmatrix}
	\left(\tilde{x}_i^{I}\right)\T&\left(\tilde{x}_i^{II}\right)\T
	\end{pmatrix}\T$,	
	then we have the following closed-loop system
	\begin{equation}\label{newdcs}
	\begin{system}{l}
	\dot{\tilde{x}}=(I\otimes A)\tilde{x}+ (I\otimes B)\sigma (u),\\
	\dot{e}=(I\otimes A-\bar{L}\otimes I)e,\\
	u=\rho (I\otimes K)(\tilde{x}-e).
	\end{system}	
	\end{equation}
		
	Then, consider the following Lyapunov function
	\begin{equation}\label{lyapunvd1}
	V=\rho\tilde{x}\T I\otimes \begin{pmatrix}
	0&0\\0&P_d
	\end{pmatrix}\tilde{x}+e\T P_De+2\int_{0}^{u}\sigma(s)ds
	\end{equation}
	where $P_d=-K_1$ and $P_D>0$ satisfies 
	\begin{equation}\label{condd3}
	P_D(I\otimes {A}-\bar{L}\otimes I)+(I\otimes {A}-\bar{L}\otimes I)\T P_D\leq -\gamma I
	\end{equation}
	with $\gamma=1+\rho\eps^{-1}\|K\|^2\|I\otimes \tilde{A}-\bar{L}\otimes I\|^2$, where $\eps$ is such that $K_2<-\frac{\eps}{2} I$ which follows from the choice of $K_2$ as negative definite matrix. Note that it can be shown that $V$ is positive definite, i.e. $V>0$ except for $(\tilde{x}, e)=0$ when $V=0$. Then, we have
	\begin{align*}
	\frac{dV}{dt}
%	=&2\tilde{x}\T I\otimes \left[\begin{pmatrix}
%	0&0\\0&P_d
%	\end{pmatrix}A\right]\tilde{x}+2\tilde{x}\T I\otimes \left[\begin{pmatrix}
%	0&0\\0&P_d
%	\end{pmatrix}B\right]\sigma(u)\\
%	&+e\T[P_D(I\otimes {A}-\bar{L}\otimes I)+(I\otimes {A}-\bar{L}\otimes I)\T P_D]e\\
%	&+2\sigma\T(u)I\otimes (KA)\tilde{x}+2\sigma\T(u)I\otimes (KB)\sigma(u)\\
%	&-2\sigma\T(u)(I\otimes K)(I\otimes {A}-\bar{L}\otimes I)e\\
	=&2\rho\sigma\T(u)I\otimes \left[KA+\begin{pmatrix}
	0&P_d
	\end{pmatrix}\right]\tilde{x}\\
	&+e\T[P_D(I\otimes {A}-\bar{L}\otimes I)+(I\otimes {A}-\bar{L}\otimes I)\T P_D]e\\
	&+\rho\sigma\T(u)I\otimes (KB+B\T K\T)\sigma(u)\\
	&-2\rho\sigma\T(u)(I\otimes K)(I\otimes {A}-\bar{L}\otimes I)e\\
	\leq&2\rho\sigma\T(u)I\otimes \left[KA+\begin{pmatrix}
	0&P_d
	\end{pmatrix}\right]\tilde{x}\\
	&-(\gamma-\rho\eps^{-1}\|K\|^2\|I\otimes \tilde{A}-\bar{L}\otimes I\|^2) \|e\|^2\\
	&+\rho\sigma\T(u)I\otimes (KB+B\T K\T+\eps I)\sigma(u)	
	\end{align*}

Meanwhile, we have
\[
KA+\begin{pmatrix}
0&P_d
\end{pmatrix}=\begin{pmatrix}
0&K_1
\end{pmatrix}+\begin{pmatrix}
0&P_d
\end{pmatrix}=0\\
%&KB+B\T K\T+\eps I=K_2+K_2\T +\eps I<0
\]
and $K_2<-\frac{\eps}{2} I$ such that
\[
\frac{dV}{dt}\leq-\|e\|^2+\rho\sigma\T(u)I\otimes (KB+B\T K\T+\eps I)\sigma(u)\leq 0
\]

Meanwhile, we can note that $\frac{dV}{dt}=0$ when $(I\otimes K)\tilde{x}=0$ and $e=0$ since \eqref{condd3}.
Thus in this case, $\tilde{x}$ is the solution of the dynamics $
\dot{\tilde{x}}^{I}=\tilde{x}^{II}$ and $\dot{\tilde{x}}^{II}=0$. And then we have $\tilde{x}^{I}=\tilde{x}^{I}(t_0)+t\tilde{x}^{II}(t_0)$ and 
$\tilde{x}^{II}=\tilde{x}^{II}(t_0)$
with $\tilde{x}^{I}=\begin{pmatrix}
\left(\tilde{x}_1^{I}\right)\T&\cdots&\left(\tilde{x}_N^{I}\right)\T
\end{pmatrix}\T$ and $\tilde{x}^{II}=\begin{pmatrix}
\left(\tilde{x}_1^{II}\right)\T&\cdots&\left(\tilde{x}_N^{II}\right)\T
\end{pmatrix}\T$, and $\tilde{x}^{I}(t_0)$ and $\tilde{x}^{II}(t_0)$ are the initial value of $\tilde{x}$ at $t_0$.

Thus, from $(I\otimes K)\tilde{x}=0$ we obtain
\begin{multline*}
(I\otimes K)\tilde{x}\\
\quad=\left[(I\otimes K_1)(\tilde{x}^{I}(t_0)+t\tilde{x}^{II}(t_0))\quad(I\otimes K_2)\tilde{x}^{II}(t_0)\right]=0
\end{multline*}
i.e. $(I\otimes K_1)(\tilde{x}^{I}(t_0)+t\tilde{x}^{II}(t_0))=0$ and $(I\otimes K_2)\tilde{x}^{II}(t_0)=0$.
Since $K_1$ and $K_2$ negative definite, we can obtain $\tilde{x}^{I}(t_0)=\tilde{x}^{II}(t_0)=0$.
Thus, the invariance set $\{(\tilde{x},e): \dot{V}(\tilde{x},e)=0\}$ contains no trajectory of the system except the trivial trajectory $(\tilde{x},e)=(0,0)$.	
Therefore, system \eqref{newdcs} is globally asymptotically stable based on LaSalle's invariance principle. 
It means we have $\tilde{x}\to 0$ and $e\to0$ when $t\to \infty$. Thus we obtain $x_i\to x_r$ as $t\to \infty$, which prove our result.
\end{proof}

\subsection{Partial-state coupling}
In this subsection we consider MAS with partial-state coupling.

\vspace{0.5cm}

\begin{Protocol}[h]	\caption{{partial-state coupling} \label{p4psc}}
The following protocol is designed for each agent
$i\in\{1,\ldots,N\}$,
\begin{equation}\label{pscpd2}
\begin{system}{cll}
\dot{\hat{x}}_i &=& A\hat{x}_i+B\hat{\zeta}_{i2}+F(\bar{\zeta}_i-C\hat{x}_i)+\iota_iB\sigma(u_i) \\
\dot{\chi}_i &=& A\chi_i+B\sigma(u_i)+\hat{x}_i-\hat{\zeta}_{i1}-\iota_{i}\chi_i \\
u_i &=&  \rho K\chi_i,
\end{system}
\end{equation}
 where $\rho>0$ is a parameter with arbitrary positive value, and $F$ is a design matrix such that $A-FC$ is Hurwitz stable.  Then, we choose matrix $K=\begin{pmatrix}
 K_1&K_2
 \end{pmatrix}$ such that $K_i \in \mathbb{R}^{m\times m}$, $i=1,2$ are arbitrary negative definite matrices, while, $\hat{\zeta}_{i1}$ and $\hat{\zeta}_{i2}$ are defined as \eqref{add_1} and \eqref{add_3}, respectively and $\bar{\zeta}_i$ is defined via \eqref{zeta_l-n}.
\end{Protocol}

\vspace{0.5cm}

 We have the following theorem for scalable global regulated state synchronization problem based on localized information exchange for MAS with partial-state coupling and double-integrator agent models.
%Then, we obtain the synchronization result for MAS with agent models mixed of double-integrators, single-integrators and neutrally stable dynamics as the following theorem. 
\begin{theorem}\label{mainthm4}
	Consider a MAS described by \eqref{eq1}, with \eqref{doubleAB} and $(A,C)$ observable, and the associated exosystem
	\eqref{solu-cond}. Let a set of nodes $\mathscr{C}$ be given which
	defines the set $\mathbb{G}_{\mathscr{C}}^N$. Let the associated
	network communication be given by \eqref{zeta_l-n}. 
	
	Then, the scalable global regulated state synchronization problem based on localized information exchange as stated in Problem
	\ref{prob2} is solvable. In particular, for any given $\rho>0$, the
	dynamic protocol \eqref{pscpd2} solves the scalable regulated state
	synchronization problem for any $N$ and any graph
	$\mathcal{G}\in\mathbb{G}_{\mathscr{C}}^N$. 
\end{theorem}

\begin{proof}[Proof of Theorem \ref{mainthm4}]
	Similar to Theorem \ref{mainthm3}, by defining $\tilde{x}_i=x_i-x_r$, $e=\tilde{x}-\chi$, and $\bar{e}=(\bar{L}\otimes I)\tilde{x}-\hat{x}$, we have the matrix expression of closed-loop system 
	\begin{equation}\label{newsysted}
	\begin{system*}{l}
	\dot{\tilde{x}}=(I\otimes A)\tilde{x}+ (I\otimes B)\sigma (u)\\
	\dot{e}=(I\otimes A-\bar{L}\otimes I)e+\bar{e}\\
	\dot{\bar{e}}=I\otimes (A-FC)\bar{e}\\
	u= \rho(I\otimes K) (\tilde{x}-e)
	\end{system*}
	\end{equation}
	
		Then we choose the following Lyapunov function:
	\begin{equation}\label{lyapunvd2}
	V= \rho\tilde{x}\T I\otimes \begin{pmatrix}
	0&0\\0&P_d
	\end{pmatrix}\tilde{x}+\begin{pmatrix}
	e\\\bar{e}
	\end{pmatrix}\T P_D\begin{pmatrix}
	e\\\bar{e}
	\end{pmatrix}+2\int_{0}^{u}\sigma(s)ds
	\end{equation}
	where  $P_d=-K_1$ and $P_D>0$ satisfies 
	\begin{equation}\label{condd4}
	P_D\bar{A}+\bar{A}\T P_D\leq -\gamma I
	\end{equation}
	where $\gamma=1+\rho\eps^{-1}\|K\|^2\|I\otimes {A}-\bar{L}\otimes I\|^2$, and $\eps$ is defined in the proof of Theorem \ref{mainthm3}, and $\bar{A}$ is defined in the proof of Theorem \ref{mainthm02}.
%	\[
%	\check{A}=\begin{pmatrix}
%	I\otimes {A}-\bar{L}\otimes I&I\\0&I\otimes ({A}-F{C})
%	\end{pmatrix}.
%	\]	
	
	Similar to the proof of Theorem \ref{mainthm3}, the synchronization result can be obtained.	
\end{proof}	

\section{MAS with Mixed-case Agents}

In this section, we will consider scalable global regulated state synchronization problem via for MAS with agent models mixed-case agents, in presence of input saturation for both networks with full and partial-state coupling. In the following assumption, we consider a class of systems which are introduced in \cite{saberi-stoorvogel-sannuti-exter}.

\begin{assumption}\label{Aass2}	
	We assume that $(A, B, C)$ is controllable and observable.
	Moreover, $A$ has eigenvalue zero with geometric multiplicity $m$ and algebraic multiplicity $m+q$ with no Jordan blocks of size larger than 2 while the remaining eigenvalues are simple purely imaginary eigenvalues.	
\end{assumption}

Obviously, this class of systems includes the neutrally stable dynamics, single- and double-integrator systems.

\subsection{Full-state coupling}
In this subsection, we design dynamic protocols for each agent via the following steps stated in Protocol \ref{p5fsc}.

\begin{Protocol}[h]	\caption{{full-state coupling} \label{p5fsc}}
\begin{itemize}
	\item First, similar to \cite[Section 4.7.1]{saberi-stoorvogel-sannuti-exter}, we use the following transformation for mixed-case agent models \eqref{eq1} by using non-singular transformation matrix $\Gamma_x$,
	\begin{equation*}
	\begin{system*}{cl}
	\tilde{A}=\Gamma_x A\Gamma_x^{-1}&=\begin{pmatrix}
	A_S&0&0\\
	0&A_F&0\\
	0&0&A_\omega
	\end{pmatrix}, \qquad\tilde{B}=\Gamma_x B=\begin{pmatrix}
	B_S\\
	B_F\\
	B_\omega
	\end{pmatrix}, \\&\tilde{C}=C\Gamma_x^{-1}=\begin{pmatrix}
	C_S&
	C_F&
	C_\omega
	\end{pmatrix}
	\end{system*}
	\end{equation*}
	where
	\[
	A_S=\begin{pmatrix}
	0&I\\0&0
	\end{pmatrix},\qquad A_F=0, \qquad A_\omega+A_\omega\T=0.
	\]
	
	\item We choose matrix $K$ so that 
	\begin{equation}
	K\tilde{A}+ \tilde{B}\T \Lambda=0 \label{cond1}
\end{equation}
	\begin{equation}
K\tilde{B}+\tilde{B}\T K\T <0 \label{cond2}
\end{equation}
	with 
	\[
	\Lambda=\begin{pmatrix}
	\Lambda_0&0&0\\0&0&0\\0&0&I
	\end{pmatrix}\text{ and }\Lambda_0=\begin{pmatrix}
	0&0\\0&P_d
	\end{pmatrix}
	\]
%	where $\eps>0$ is a sufficiently small positive number and
	where $P_d>0$ is any positive definite matrix. The existence of matrix $K$ is proved in \cite[Page 235]{saberi-stoorvogel-sannuti-exter}.
\item Next, the following protocol is designed for each agent
$i\in\{1,\ldots,N\}$,
\begin{equation}\label{pscpm1}
\begin{system}{cll}
\dot{\chi}_i &=& A\chi_i+B\sigma(u_i)+\bar{\zeta}_i-\hat{\zeta}_i-\iota_i\chi_i \\
u_i &=&\rho K\Gamma_x\chi_i,
\end{system}
\end{equation}
where $\rho>0$ is a parameter with arbitrary positive value, $\hat{\zeta}_i$ and $\bar{\zeta}_i$ are defined by \eqref{info1} and \eqref{zeta_l-n}, respectively.
\end{itemize}
\end{Protocol}
\pagebreak

 We have the following theorem for scalable global regulated state synchronization problem based on localized information exchange for MAS with full-state coupling and mixed-case agent models.
\begin{theorem}\label{mainthm5}
	Consider a MAS described by \eqref{eq1} with $C=I$ satisfying Assumption \ref{Aass2}, and the associated exosystem
	\eqref{solu-cond}. Let a set of nodes $\mathscr{C}$ be given which
	defines the set $\mathbb{G}_{\mathscr{C}}^N$. Let the associated
	network communication be given by \eqref{zeta_l-n}. 
	
	Then, the scalable global regulated state synchronization problem based on localized information exchange for MAS with full-state coupling as stated in Problem
	\ref{prob2} is solvable. In particular, for any given $\rho>0$, the
	dynamic protocol \eqref{pscpm1} with \eqref{cond1} and \eqref{cond2} solves the regulated state
	synchronization problem for any $N$ and any graph
	$\mathcal{G}\in\mathbb{G}_{\mathscr{C}}^N$. 
\end{theorem}
\begin{proof}[Proof of Theorem \ref{mainthm5}]
	Firstly, we have
	\begin{equation*}
	\begin{system*}{l}
	\dot{\tilde{x}}_i=A\tilde{x}_i+ B\sigma (u_i),\\
	\dot{e}_i=Ae_i-\sum_{j=1}^{N}\bar{\ell}_{ij}e_j,\\
	u_i=\rho K\Gamma_x(\tilde{x}_i-e_i)
	\end{system*}	
	\end{equation*}
	by $\tilde{x}_i=x_i-x_r$ and $e_i=\tilde{x}_i-\chi_i$. Then, let
	\[
	\tilde{x}=\begin{pmatrix}
	\tilde{x}_1\\\vdots\\\tilde{x}_N
	\end{pmatrix} ,
	u=\begin{pmatrix}
	u_1\\\vdots\\u_N
	\end{pmatrix}, 
	e=\begin{pmatrix}
	e_1\\\vdots\\e_N
	\end{pmatrix}, \text{ and } 
	\sigma(u)=\begin{pmatrix}
	\sigma(u_1)\\\vdots\\\sigma(u_N)
	\end{pmatrix}
	\]
	then we have the following closed-loop system
	\begin{equation}\label{newcs}
	\begin{system}{l}
	\dot{\tilde{x}}=(I\otimes A)\tilde{x}+ (I\otimes B)\sigma (u),\\
	\dot{e}=(I\otimes A-\bar{L}\otimes I)e,\\
	u=\rho I\otimes K\Gamma_x(\tilde{x}-e).
	\end{system}	
	\end{equation}
	
	We transform mixed-case agent model \eqref{newcs} as
	\begin{equation}\label{newcs2}
	\begin{system}{l}
	\dot{\eta}=\begin{pmatrix}
	I\otimes \tilde{A}&0\\0&I\otimes \tilde{A}-\bar{L}\otimes I
	\end{pmatrix}\eta+\begin{pmatrix}
	I\otimes \tilde{B}\\0
	\end{pmatrix}\sigma \left(u\right)\\
	u=\rho I\otimes \begin{pmatrix}
	K&-K
	\end{pmatrix}\eta
	\end{system}
	\end{equation}
	by a non-singular matrix $I\otimes \Gamma_x$, where $\eta=\begin{pmatrix}
	{\eta}_x\T&{\eta}_e\T
	\end{pmatrix}\T=\begin{pmatrix}
	(I\otimes \Gamma_x\T) \tilde{x}\T& (I\otimes \Gamma_x\T)e\T
	\end{pmatrix}\T$.

	Next, we choose the following Lyapunov function:
	\begin{equation}\label{lyapunv1}
	V=\eta\T \begin{pmatrix}
	\rho I\otimes\Lambda&0\\0&P_0
	\end{pmatrix}\eta+2\int_{0}^{u}\sigma(s)ds
	\end{equation}
	where $P_0>0$ satisfies 
	\begin{equation}\label{cond3}
	P_0(I\otimes \tilde{A}-\bar{L}\otimes I)+(I\otimes \tilde{A}-\bar{L}\otimes I)\T P_0\leq -\gamma I
	\end{equation}
	with $\gamma=1+\rho\eps^{-1}\|K\|^2\|I\otimes \tilde{A}-\bar{L}\otimes I\|^2$, where $\eps$ is such that $K\tilde{B}+\tilde{B}\T K\T <-\eps I$, note that \eqref{cond2} guarantees existence of $\eps$.  It can be shown that $V$ is positive definite, i.e. $V>0$ except for $(\tilde{x}, e)=0$ when $V=0$. Then, we have
	\begin{align*}
	\nonumber\frac{dV}{dt}=&
	2\eta\T \begin{pmatrix}
	\rho I\otimes(\Lambda \tilde{A}) &0\\0&P_0(I\otimes \tilde{A}-\bar{L}\otimes I)
	\end{pmatrix} \eta\\
	\nonumber& +2\rho\eta\T \begin{pmatrix}
	I\otimes(\Lambda \tilde{B})\\0
	\end{pmatrix}\sigma \left(u\right)\\
	\nonumber&+2\rho \sigma\T\left(u\right) \begin{pmatrix}
	I\otimes (K\tilde{A})&-(I\otimes K\tilde{A}-\bar{L}\otimes K)
	\end{pmatrix}\eta\\
	\nonumber&+2\rho\sigma\T\left(u\right)I\otimes (K\tilde{B})\sigma\left(u\right)\\
	\leq & -\gamma\eta_e\T \eta_e+2\rho\sigma\T(u)(I\otimes K\tilde{B})\sigma(u)\\
	&-2\rho\sigma\T(u)(I\otimes K)(I\otimes\tilde{A}-\bar{L}\otimes I)\eta_e\\
	%\leq &-\gamma\eta_e\T \eta_e+2\rho\sigma\T(u)(I\otimes K\tilde{B})\sigma(u)\\
	%&-2\rho\sigma\T(u)\left[(I\otimes K)(I\otimes\tilde{A}-\bar{L}\otimes I)\eta_e\right]\\
	\leq&-\gamma\eta_e\T \eta_e+\eps^{-1}\rho\|K\|^2\|I\otimes\tilde{A}-\bar{L}\otimes I\|^2\eta_e\T\eta_e\\
	&+\rho\sigma\T(u)(I\otimes(K\tilde{B}+\tilde{B}\T K\T+\eps I))\sigma(u)\\
	=&-\|\eta_e\|^2+\rho\sigma\T(u)\left[I\otimes(K\tilde{B}+\tilde{B}\T K\T+\eps I)\right]\sigma(u)
	%\label{lypunbo2}
	\end{align*}
	since we have \eqref{cond1} and \eqref{cond3}. Because $(\tilde{A}, \tilde{B})$ is surjective, we have a solution $K$ such that
	$
	\frac{dV}{dt}\leq  0
	$
	provided \eqref{cond2}. 
	
	Then, we note that the $\frac{dV}{dt}=0$ when $(I\otimes K)\eta_x=0$ and $\eta_e=0$, the
	dynamics of \eqref{newcs2} becomes $\dot{\eta}_x=(I\otimes \tilde{A})\eta_x$.
		
	Similar to the proof of \cite[Theorem 4.61]{saberi-stoorvogel-sannuti-exter} with \eqref{cond1} and \eqref{cond2}, we can obtain $(I\otimes K)\eta_x=0$ only when $\eta_x(t_0)=0$.
	
	Thus, we obtain the global asymptotic stability of the closed-loop system \eqref{newcs2}, i.e. we have $\eta_x\to 0$. It implies $\tilde{x}\to 0$ since $(I
	\otimes\Gamma_x^{-1})\eta_x\to 0$ when $t\to \infty$, and thus we have $x_i\to x_r$ as $t\to \infty$, which prove our result.
\end{proof}

\subsection{Partial-state coupling}
In this subsection we consider MAS with partial-state coupling.

\vspace{0.5cm}

\begin{Protocol}[h]	\caption{{partial-state coupling} \label{p6psc}}
The following protocol is designed for each agent
$i\in\{1,\ldots,N\}$,
\begin{equation}\label{pscpm2}
\begin{system}{cll}
\dot{\hat{x}}_i &=& A\hat{x}_i+B\hat{\zeta}_{i2}+F(\bar{\zeta}_i-C\hat{x}_i)+\iota_iB\sigma(u_i) \\
\dot{\chi}_i &=& A\chi_i+B\sigma(u_i)+\hat{x}_i-\hat{\zeta}_{i1}-\iota_{i}\chi_i \\
u_i &=& \rho K\Gamma_x\chi_i,
\end{system}
\end{equation}
where $F$ is a design matrix such that $A-FC$ is Hurwitz stable, 
$\Gamma_x$ is a non-singular matrix, $\rho>0$ is a parameter with arbitrary positive value, and $K$ satisfies \eqref{cond1} and \eqref{cond2}, where $\hat{\zeta}_{i1}$ and $\hat{\zeta}_{i2}$ are defined as \eqref{add_1} and \eqref{add_3}, respectively and $\bar{\zeta}_i$ is defined via \eqref{zeta_l-n}.
\end{Protocol}

\vspace{0.5cm}

 We have the following theorem for scalable global regulated state synchronization problem based on localized information exchange for MAS with partial-state coupling and mixed-case agent models.

%Meanwhile, we have
%\[
%\tilde{A}-\Gamma_x F\tilde{C}=\Gamma_x A \Gamma_x^{-1}-\Gamma_x AC \Gamma_x^{-1} =\Gamma_x({A}-F{C})\Gamma_x^{-1}
%\]
%is also Hurwitz stable.
%Then, we obtain the synchronization result for mixed-case agent as the following theorem. 
\begin{theorem}\label{mainthm6}
	Consider a MAS described by \eqref{eq1} satisfying Assumption \ref{Aass2}, and the associated exosystem
	\eqref{solu-cond}. Let a set of nodes $\mathscr{C}$ be given which
	defines the set $\mathbb{G}_{\mathscr{C}}^N$. Let the associated
	network communication be given by \eqref{zeta_l-n}. 
	
	Then, the scalable global regulated state synchronization problem based on localized information exchange as stated in Problem
	\ref{prob2} is solvable. In particular, for any given $\rho>0$, the
	dynamic protocol \eqref{pscpm2} with \eqref{cond1} and \eqref{cond2} solves the scalable regulated state
	synchronization problem for any $N$ and any graph
	$\mathcal{G}\in\mathbb{G}_{\mathscr{C}}^N$. 
\end{theorem}

\begin{proof}[Proof of Theorem \ref{mainthm6}]
	Similar to Theorem \ref{mainthm3}, by defining $\tilde{x}_i=x_i-x_r$, $e=\tilde{x}-\chi$, and $\bar{e}=(\bar{L}\otimes I)\tilde{x}-\hat{x}$, we have the matrix expression of closed-loop system 
	\begin{equation}\label{newsystemm}
	\begin{system*}{l}
	\dot{\tilde{x}}=(I\otimes A)\tilde{x}+ (I\otimes B)\sigma (u)\\
	\dot{e}=(I\otimes A-\bar{L}\otimes I)e+\bar{e}\\
	\dot{\bar{e}}=I\otimes (A-FC)\bar{e}\\
	u= \rho(I\otimes K\Gamma_x) (\tilde{x}-e)
	\end{system*}
	\end{equation}

	Then, by using nonsingular matrix $I\otimes \Gamma_x$, we can obtain 
	\begin{equation}
	\begin{system*}{l}
	\dot{\eta}_x=(I\otimes \tilde{A})\eta_x+ (I\otimes \tilde{B})\sigma (u)\\
	\dot{\eta}_e=(I\otimes \tilde{A}-\bar{L}\otimes I)\eta_e+\eta_{\bar{e}}\\
	\dot{\eta}_{\bar{e}}=I\otimes (\tilde{A}-\Gamma_x F\tilde{C})\eta_{\bar{e}}\\
	u= \rho(I\otimes K) (\eta_x-\eta_e)
	\end{system*}
	\end{equation}
	where $\eta_x=(I\otimes \Gamma_x)\tilde{x}$, $\eta_e=(I\otimes \Gamma_x) e$, and $\eta_{\bar{e}}=(I\otimes \Gamma_x){\bar{e}}$. 
		
	Then we choose the following Lyapunov function:
	\begin{equation}\label{lyapunv2}
	V= \bar{\eta}\T \begin{pmatrix}
	\rho I\otimes\Lambda&0\\0&P_0
	\end{pmatrix}\bar{\eta}+2\int_{0}^{u}\sigma(s)ds
	\end{equation}
	where $\bar{\eta}=\begin{pmatrix}
	\eta_x\T&\eta_e\T&\eta_{\bar{e}}\T
	\end{pmatrix}\T$ and $P_0>0$ satisfies 
	\begin{equation}\label{cond4}
	P_0\hat{A}+\hat{A}\T P_0\leq -\gamma I
	\end{equation}
	where $\gamma=1+\eps^{-1}\rho\|K\|^2\|I\otimes \tilde{A}-\bar{L}\otimes I\|^2$, and $\eps$ is the same as in the proof of Theorem \ref{mainthm5} and
	\[
	\hat{A}=\begin{pmatrix}
	I\otimes \tilde{A}-\bar{L}\otimes I&I\\0&I\otimes (\tilde{A}-\Gamma_x F\tilde{C})
	\end{pmatrix}.
	\]	
	
	Thus, similar to the proof of Theorem \ref{mainthm5}, the synchronization result can be obtained.
\end{proof}
\vspace*{-6mm}
\begin{remark}
	It is worth to note that in all of the protocols for MAS with neutrally stable, double-integrator, and mixed of double-integrator, single-integrator and neutrally stable dynamics, the choice of positive parameter $\rho$ is independent of the communication graph and as such it establishes infinite gain margin for our protocols.
\end{remark}

%\begin{remark}
%	In this paper, we develop a scalable linear observer-based protocol design to achieve global synchronization for MAS. However, the protocol design must use the nonlinear signal $\sigma(u_i)$, and hence it is not a purely linear design. We will focus on the purely linear protocol design for global synchronization problem as a future work.
%\end{remark}

\section{Numerical Example}
In this section, we will illustrate the effectiveness of our protocols with numerical examples for global synchronization of MAS with double-integrator and mixed-case agent models with partial-state coupling. 
\begin{figure}[t]
	\includegraphics[width=4cm, height=2.5cm]{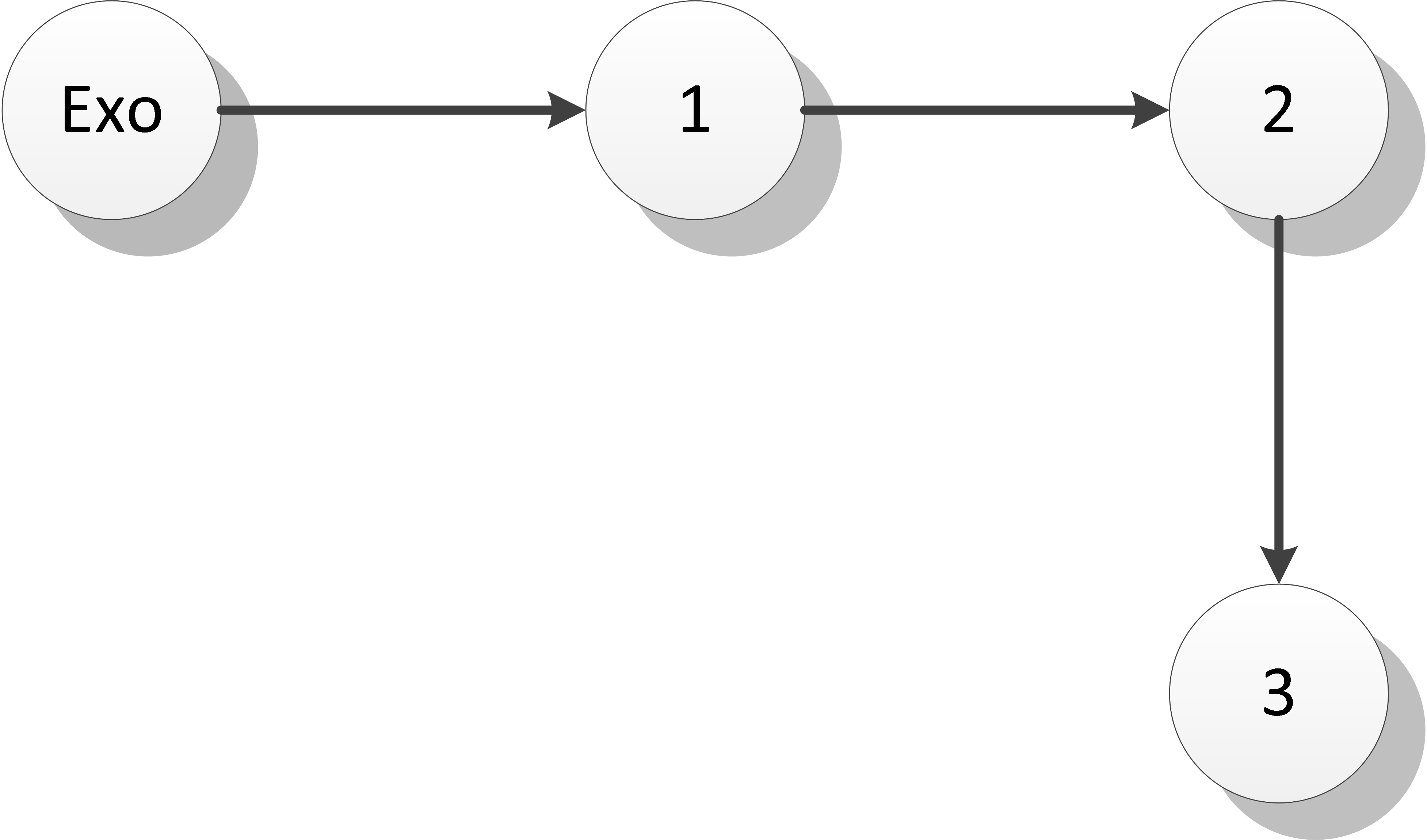}
	\centering
	\caption{The directed communication network $1$}\label{graph_3nodes}
\end{figure}
\begin{figure}[t]
	\includegraphics[width=7cm, height=2.5cm]{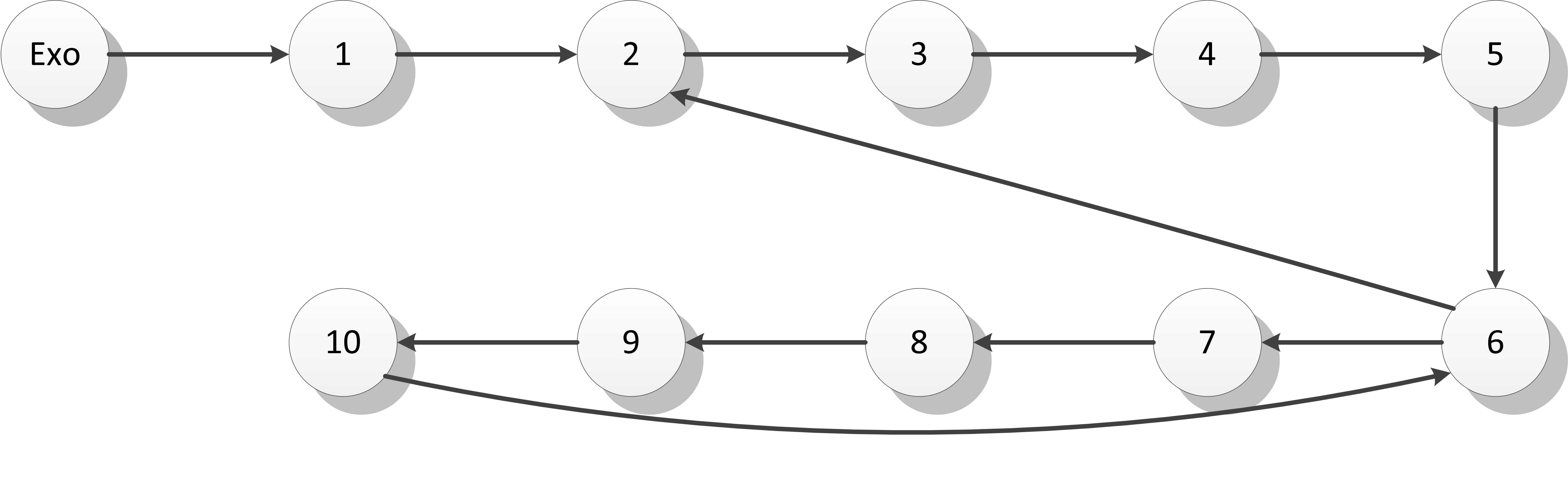}
	\centering	
	\caption{The directed communication network $2$}\label{graph_10nodes}
\end{figure}
\subsection*{Example 1: Double-integrator}
Consider a MAS with double-integrators agent models \eqref{eq1} as:
\begin{equation*}\label{ex_2}
\begin{system*}{cl}
\dot{x}_i&=\begin{pmatrix}
0&1\\0&0
\end{pmatrix}x_i+\begin{pmatrix}
0\\1
\end{pmatrix}\sigma(u_i),\\
y_i&=\begin{pmatrix}
1&0
\end{pmatrix}x_i
\end{system*}
\end{equation*}
and the exosystem:
\begin{equation*}\label{exo_ex1}
\begin{system*}{cl}
\dot{x}_r=\begin{pmatrix}
0&1\\0&0
\end{pmatrix}x_r,\quad
y_r=\begin{pmatrix}
1&0
\end{pmatrix}x_r
\end{system*}
\end{equation*}

By choosing parameter $\rho=1$ and matrices $F$ and $K$ as
\begin{equation*}
F=\begin{pmatrix}
1\\2
\end{pmatrix}, \quad K=\begin{pmatrix}
-10&-2
\end{pmatrix}
\end{equation*}
the scalable Protocol \ref{p4psc} would be equal to 
\begin{equation}\label{pscpd2-Ex}
\begin{system}{cll}
\dot{\hat{x}}_i &=& \begin{pmatrix}
-1&1\\-2&0
\end{pmatrix}\hat{x}_i+\begin{pmatrix}
0\\1
\end{pmatrix}\hat{\zeta}_{i2}+\begin{pmatrix}
1\\2
\end{pmatrix}\bar{\zeta}_i+\iota_i\begin{pmatrix}
0\\1
\end{pmatrix}\sigma(u_i) \\
\dot{\chi}_i &=& \begin{pmatrix}
0&1\\0&0
\end{pmatrix}\chi_i+\begin{pmatrix}
0\\1
\end{pmatrix}\sigma(u_i)+\hat{x}_i-\hat{\zeta}_{i1}-\iota_{i}\chi_i \\
u_i &=&  \begin{pmatrix}
-10&-2
\end{pmatrix}\chi_i,
\end{system}
\end{equation}
where $\iota_1=1$ and $\iota_{i}=0$ for $i=\{1,\hdots, N\}$.
First, consider a MAS with $3$ nodes and communication graph as Figure \ref{graph_3nodes}. 

To illustrate the scalibility of our protocols we show that the designed protocol will also work for MAS with $10$ nodes with communication topology as Figure \ref{graph_10nodes}.

 The simulation results are shown in Figure \ref{results_double_case3Nodes} and Figure \ref{results_double_case10Nodes} for MAS with $3$ and $10$ agents, respectively.

\subsection*{Example 2: Mixed-case }
 In this example, we consider MAS with mixed-case agent model which contains two double-integrator, one single-integrator and neutrally stable dynamics as:
\begin{equation*}\label{ex}
\begin{system*}{cl}
\dot{x}_i&=\begin{pmatrix}
0&0&1&0&0&0&0\\0&0&0&1&0&0&0\\0&0&0&0&0&0&0\\0&0&0&0&0&0&0\\0&0&0&0&0&0&0\\0&0&0&0&0&0&1\\0&0&0&0&0&-1&0
\end{pmatrix}x_i+ \begin{pmatrix}
0&1&3\\0&0&5\\1&2&4\\0&1&6\\0&0&1\\1&1&0\\1&0&1
\end{pmatrix}\sigma(u_i)\\
y_i&= \begin{pmatrix}
1&1&1&1&1&1&1\\1&0&0&0&0&0&0\\0&0&0&0&0&0&1\\0&0&1&1&1&1&1
\end{pmatrix}x_i
\end{system*}
\end{equation*}

and the associated exosystem:
\begin{equation*}\label{exo_ex2}
\begin{system*}{cl}
\dot{x}_r=\begin{pmatrix}
0&0&1&0&0&0&0\\0&0&0&1&0&0&0\\0&0&0&0&0&0&0\\0&0&0&0&0&0&0\\0&0&0&0&0&0&0\\0&0&0&0&0&0&1\\0&0&0&0&0&-1&0
\end{pmatrix}x_r,\quad
y_r= \begin{pmatrix}
1&1&1&1&1&1&1\\1&0&0&0&0&0&0\\0&0&0&0&0&0&1\\0&0&1&1&1&1&1
\end{pmatrix}x_r
\end{system*}
\end{equation*}
 \begin{figure}[t]
	\includegraphics[width=8cm, height=6.5cm]{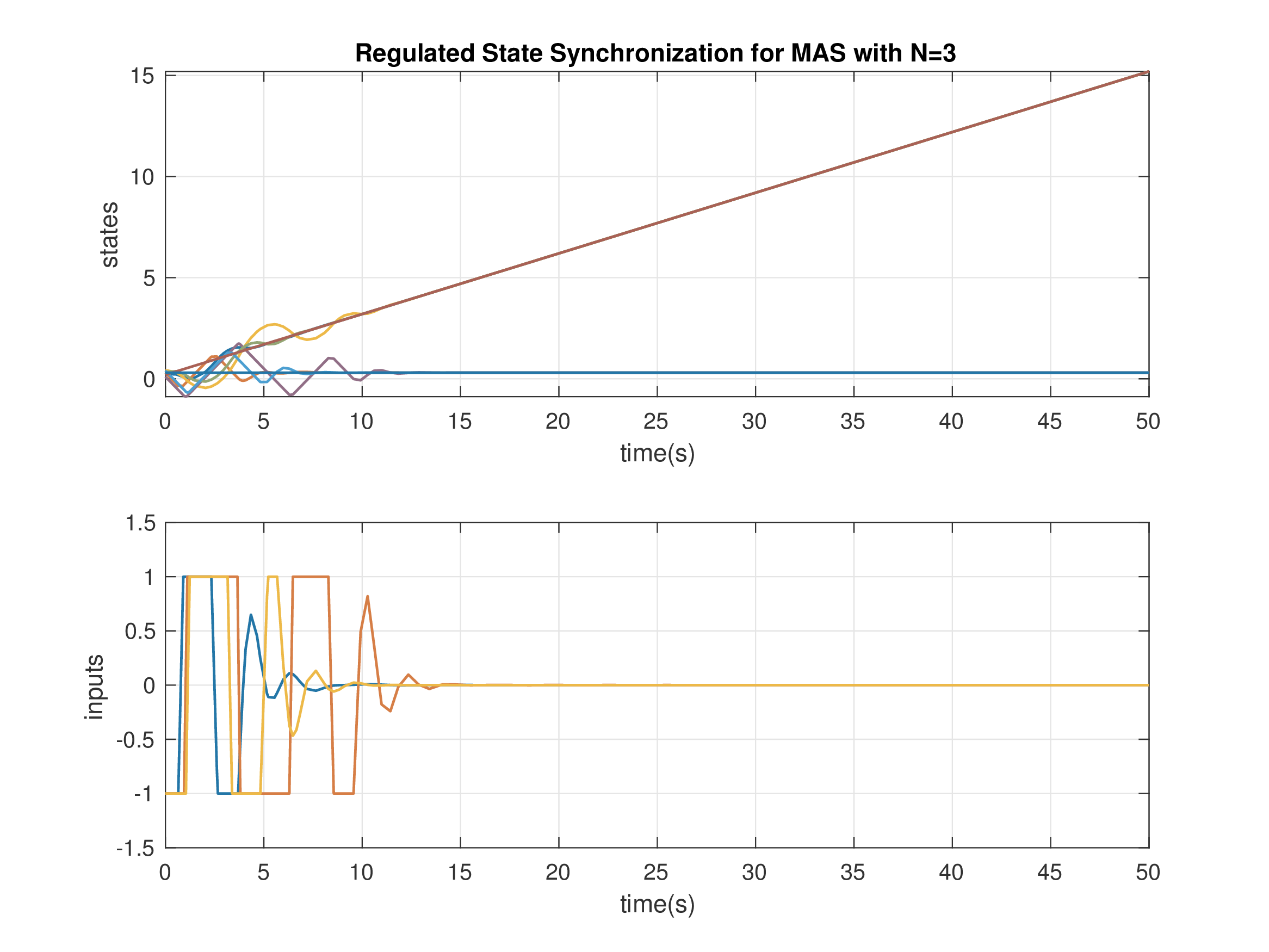}
	\centering
	\caption{Regulated state synchronization for MAS with double-integrator agents, partial-state coupling and $3$ agents}\centering \label{results_double_case3Nodes}
\end{figure}
\begin{figure}[t]
	\includegraphics[width=8cm, height=6.5cm]{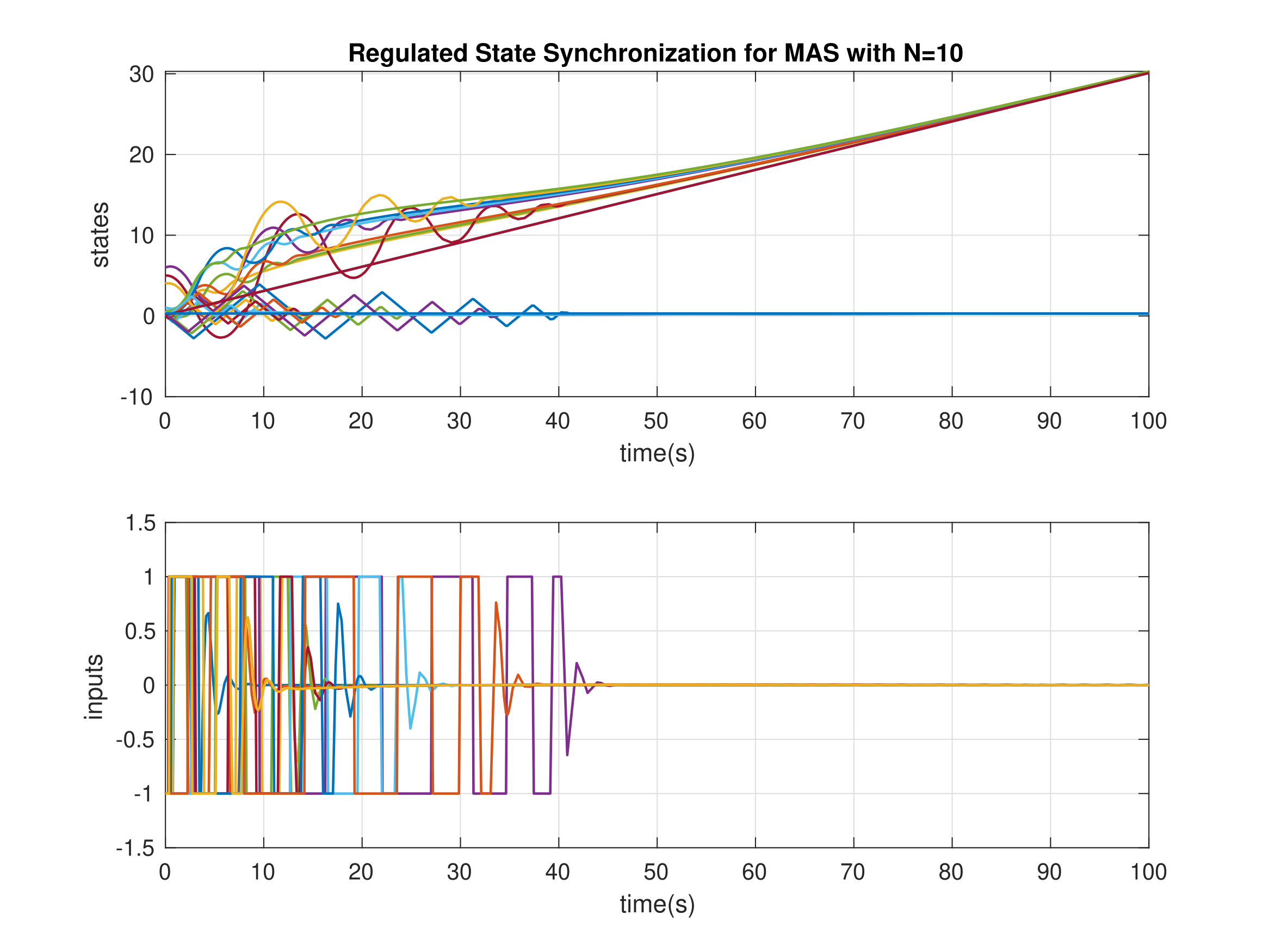}
	\centering
	\caption{Regulated state synchronization for MAS with double-integrator agents, partial-state coupling and $10$ agents}\label{results_double_case10Nodes}
\end{figure}
\begin{figure}[t]
	\includegraphics[width=8cm, height=6.5cm]{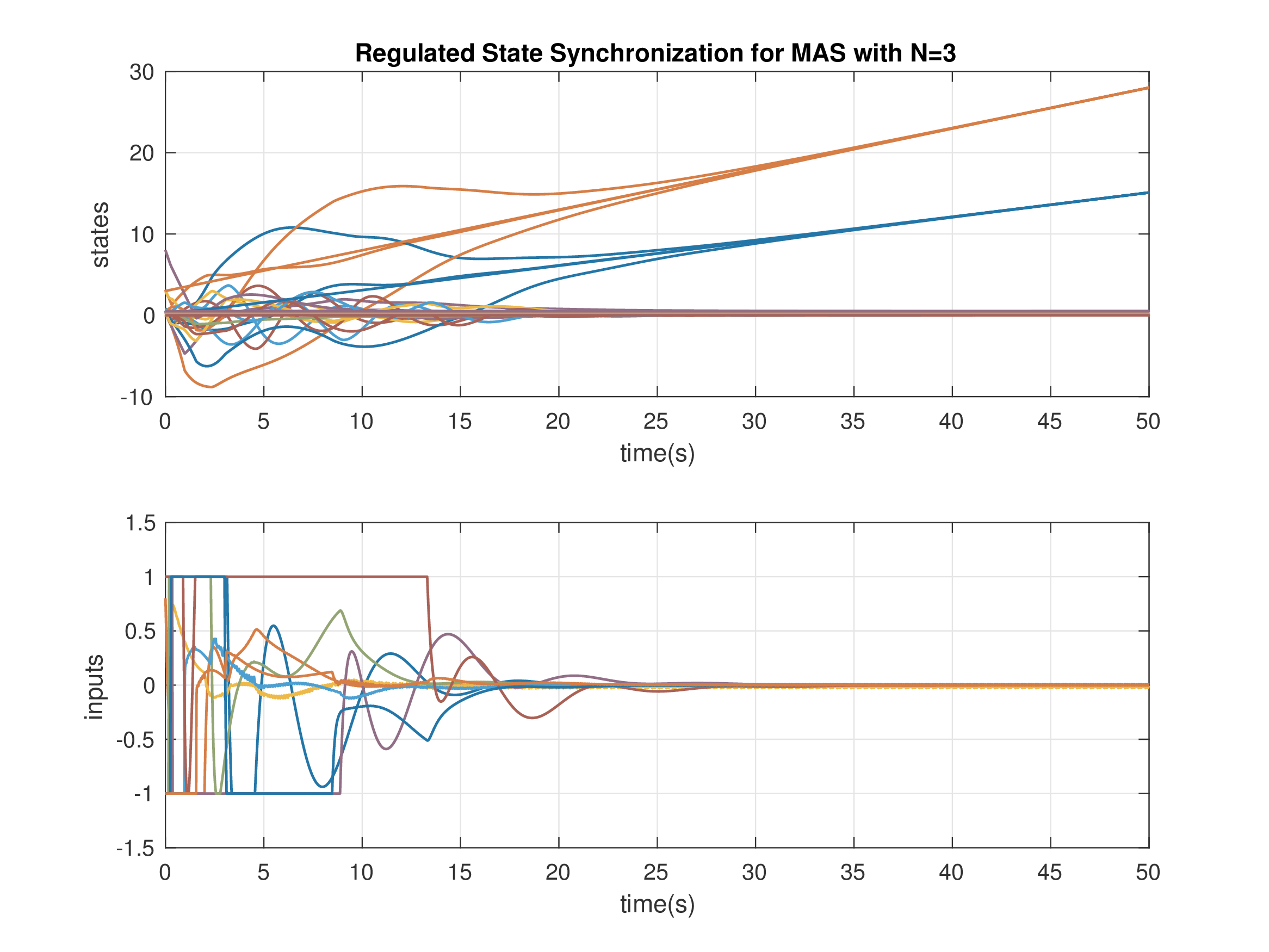}
	\centering
	\caption{Regulated state synchronization for MAS with mixed-case agents, partial-state coupling and $3$ agents}\label{results_mixed_case3Nodes}
\end{figure}
\begin{figure}[t]
	\includegraphics[width=8cm, height=6.5cm]{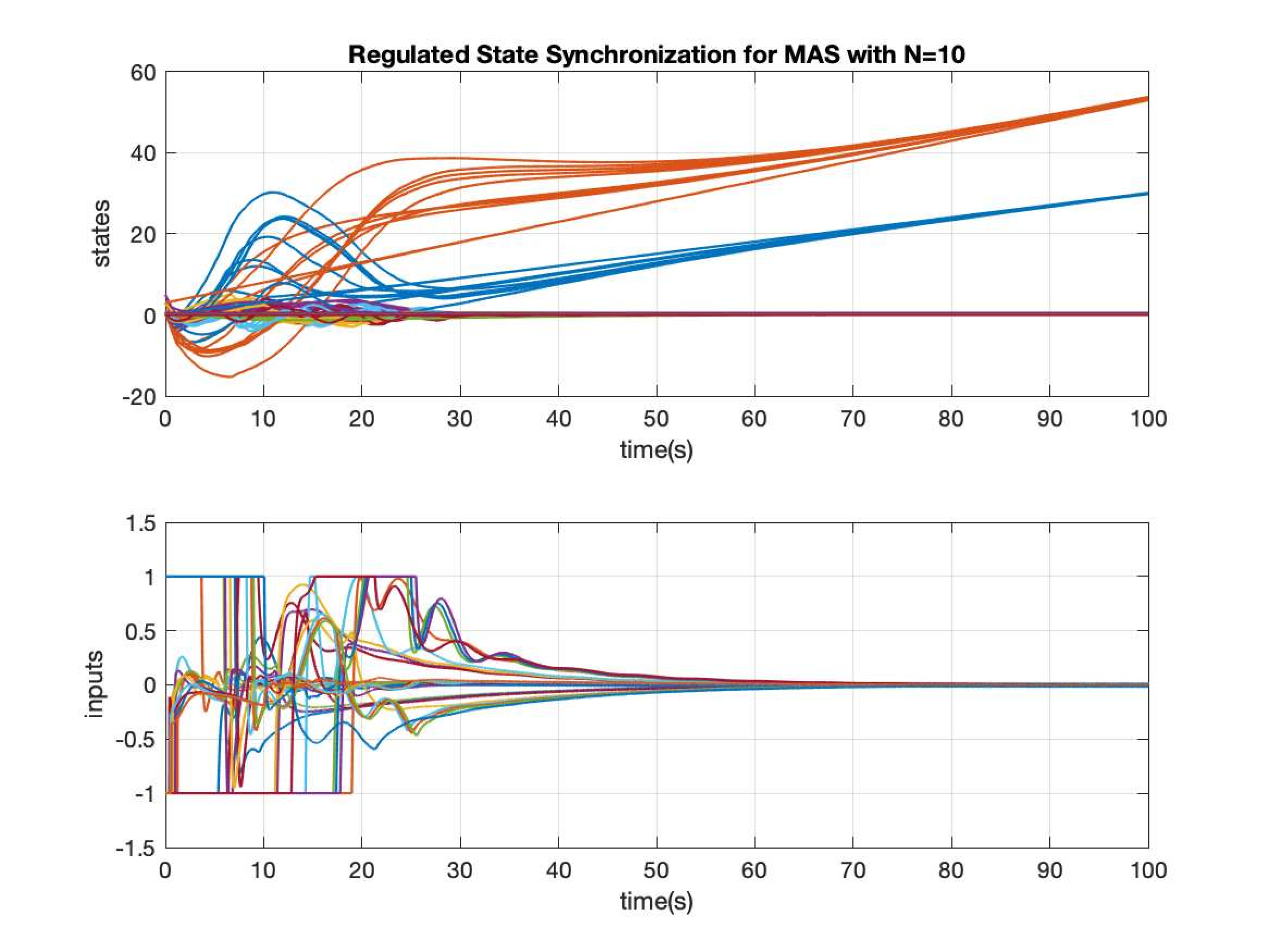}
	\centering
	\caption{Regulated state synchronization for MAS with mixed-case agents, partial-state coupling and $10$ agents}\label{results_mixed_case10Nodes}
\end{figure}

We choose parameter $\rho=1$ and matrix $K$ and $F$ as following:
\begin{align*}
&F= \begin{pmatrix}
0.55&6.81&0.73&-0.42\\
7.97&-7.41 &1.30&-8.30\\
0.57&10&2.97&0.37\\
11.14&-10.32&5.06&-11.24\\
-5.92&-0.92&3.66&7.89\\
-7.01&1.98&-14.49&8.53\\
1.35&-0.27&8.48&-1.52
\end{pmatrix}\\
&K=\begin{pmatrix}
-1&0&-4&6&-22&-1&1\\
 -2&-1&-3&-2&18&0&1\\
-4&-6&-5&-3&-61&-1&0
\end{pmatrix}
\end{align*}

Consider a MAS  with $3$ agents, and associated directed communication topology shown in Figure \ref{graph_3nodes}.

The simulation results for global state synchronization of the MAS with partial-state coupling via scalable dynamic protocol \eqref{pscpm2} are shown in Figure \ref{results_mixed_case3Nodes}.

 To show the scalability of our protocol designs, we consider a MAS with $10$ nodes and agent models as the previous case with communication topology as Figure \ref{graph_10nodes}. 
 
 The simulation results shown in Figure \ref{results_mixed_case10Nodes} show that global state synchronization is achieved with the same designed protocol.

\bibliographystyle{plain}
\bibliography{referenc}
\end{document}